\documentclass[lettersize,journal]{IEEEtran}
\usepackage{algorithm}
\usepackage{algorithmic}
\usepackage{array}
\usepackage[caption=false,font=normalsize,labelfont=sf,textfont=sf]{subfig}
\usepackage{textcomp}
\usepackage{stfloats}
\usepackage{url}
\usepackage{verbatim}
\usepackage{graphicx}
\usepackage[nospace]{cite}
\usepackage{amsfonts}
\usepackage[fleqn]{amsmath}
\usepackage{array}
\usepackage[caption=false,font=normalsize,labelfont=sf,textfont=sf]{subfig}
\usepackage{textcomp}
\usepackage{stfloats}
\usepackage{multirow}
\usepackage{multicol}
\usepackage{url}
\usepackage{verbatim}
\usepackage{graphicx}
\usepackage{rotating}
\usepackage{tabularray}
\usepackage{pbox}
\usepackage{makecell}
\usepackage{footnote}
\usepackage{nicematrix}
\usepackage{booktabs}     
\usepackage{graphicx}     
\usepackage{booktabs}
\usepackage{siunitx}
\usepackage{multirow}
\usepackage{comment}

\usepackage{array}
\allowdisplaybreaks
\newtheorem{prop}{Proposition}

\newtheorem{proof}{Proof}

\begin{document}
\bstctlcite{BSTcontrol}
\title{Contracting for Long-Duration Energy Storage in Incomplete Risk Markets}

\author{{Adam Suski, Elina Spyrou, Jacob Mays, Richard Green}
\thanks{Adam Suski and Elina Spyrou are with the Department of Electrical and Electronic Engineering, Imperial College London. (e-mail: a.suski23@imperial.ac.uk). Jacob Mays is with the  School of Civil and Environmental Engineering, Cornell University. Richard Green is with the Department of Economics \& Public Policy, Imperial Business School.}
\thanks{This work was funded by the Taylor Donation from the Grantham Institute and Energy Futures Lab, and Global Fellows Fund from Imperial College London,  the Leverhulme International Professorship with grant reference LIP-2020-002, and the Engineering and Physical Sciences Research Council under the grant EP/Y025946/1 (Electric Power Innovation for a Carbon-free Society (EPICS)).}}

\maketitle

\begin{abstract}
Long-duration energy storage (LDES) is increasingly regarded as essential for reliability in decarbonized power systems. To encourage investment, policymakers introduce contracts, such as cap-and-floor schemes. So far, these schemes have only been evaluated using exogenous revenue or price distributions. This paper develops a two-stage stochastic equilibrium model to evaluate how LDES cap-and-floor design affects investment and market outcomes. This model endogenously captures the interactions among contract design, investment capacity, and cost of capital. Results for a stylized Great Britain case study show that market incompleteness substantially suppresses LDES investment. Centrally administered zero-premium contracts can restore the risk-neutral investment level by reducing downside risk, but doing so requires substantial expected transfers from consumers to investors and produces outcomes that are sensitive to the cap, floor, and sharing parameters. Bilaterally negotiated contracts largely eliminate expected transfers and reduce sensitivity to those parameters, but provide weaker investment incentives. 

To balance investment incentives, transfers, and social welfare, policymakers should jointly consider contract and institutional design. 

\end{abstract}


\section{Introduction
}
\IEEEPARstart{L}{ong}-duration energy storage (LDES) is increasingly recognized as a cornerstone of deeply decarbonized power systems \cite{kittel2026long,AFRY_EURELECTRIC,levin2026deployment}.
LDES can provide a range of valuable system services. In particular, it is uniquely suited to applications requiring long discharge durations, such as intertemporal energy arbitrage across days and the provision of resource adequacy services \cite{cheng_electricity_2025}. As a result, the investment case for LDES is expected to rely primarily on wholesale market revenues from medium- to long-term energy arbitrage, complemented by a smaller but relatively stable revenue stream from resource adequacy services \cite{LCP_REGEN_SCENARIO}.

Although simulations suggest that LDES projects recover their costs in expectation, their revenue distributions are highly skewed \cite{suski_missing_2025, billimoria2026caps }. In the absence of sufficiently liquid markets for long-term hedging instruments, this interannual variability in revenues can substantially increase the cost of capital and deter investment \cite{makrides_quantifying_2025}. Markets for long-term hedging instruments remain limited, particularly in Europe and the UK \cite{newbery_missing_2016, Battle2023}.

To address this market incompleteness, several policymakers have begun introducing technology-specific contract-based mechanisms, consistent with the hybrid market proposals of Roques et al. \cite{roques_adapting_2017} and Joskow \cite{joskow_hierarchies_2022}. In the UK, this has taken the form of a Cap-and-Floor (C\&F) scheme for net revenues \cite{department_for_energy_security__net_zero_long_2024}, modeled after the regulatory framework for electricity interconnectors \cite{ofgem_cap_2016}. Australia has initiated competitive tenders under the Long-Term Energy Service Agreement (LTESA) model \cite{nsw_department_of_planning_industry_and_environment_long-term_2021}, and similar programs are emerging in New York \cite{new_york_state_energy_research_and_development_authority_bulk_2026}. Although these mechanisms differ in structure, they share a common goal: facilitating risk sharing between consumers and investors to ensure that socially beneficial capacity investments are realized. 

Despite the growing adoption of such mechanisms, quantitative evidence remains limited regarding three questions central to their design. First, how effective is the C\&F mechanism in reducing investment risk and stimulating LDES deployment? Second, how are the resulting costs and benefits distributed between consumers and investors? Third, which contract design choices best balance investment incentives, consumer outcomes, and system efficiency?\cite{spyrou_designing_2026}.

Existing studies provide only partial answers. Assessments of system-wide benefits have largely relied on central planner models, which estimate the value of additional LDES capacity but do not capture investor behavior under risk \cite{afry_benefits_2022,LCP_REGEN_SCENARIO}. 

A growing literature has investigated contract-based support mechanisms for storage under uncertainty. Mays and Jenkins \cite{mays_financial_2023} show that revenue put options can reduce the implied cost of capital for short-duration storage. Berdin et al. \cite{Berdin_2026} demonstrate that capacity remuneration mechanisms can restore investment incentives for long-duration hydrogen storage, particularly when combined with renewable support policies. Suski et al. \cite{suski_comparing_2026} compare several storage support mechanisms, including a floor-only contract.\footnote{Although the paper refers to a Cap-and-Floor contract, the cap is so high that the simulated mechanism is effectively equivalent to a floor-only contract.} 

Turning to the literature on Cap-and-Floor contract design, existing work employs numerous scenarios to simulate gross margins or prices \cite{cepa_cap_2025, billimoria2026caps}. This overlooks an important feedback mechanism: contract design shapes investment decisions, which in turn influence market outcomes and thus the revenues that the contracts are intended to insure.

Another overlooked design dimension concerns the institutional allocation of contracting responsibilities.  In Great Britain, Cap-and-Floor contracts are administered centrally by the energy regulator, with candidate projects selected with input from the National Energy System Operator. In other jurisdictions, however, long-term contracts are often procured through decentralized competitive auctions in which suppliers negotiate directly with project developers on behalf of consumers \cite{moreno_auction_2010}. Understanding whether decentralized contracting can replicate the outcomes of a centrally administered mechanism is therefore an important policy question.

To address these gaps in the literature, we construct a two-stage stochastic equilibrium model of investment under risk market incompleteness, building on the framework originally developed by Mays et al. \cite{mays_asymmetric_2019}. The model jointly determines investment, electricity market equilibrium, and cap-and-floor contract pricing under market incompleteness. We apply the model for a stylized future power system in Great Britain.

The framework enables us to address the three questions posed above. First, we quantify the extent to which Cap-and-Floor contracts mitigate investment risk and stimulate LDES deployment. Second, we compare the distribution of costs and benefits under two institutional arrangements: a centrally administered scheme offering zero-premium contracts and a decentralized market in which consumers negotiate contracts competitively with LDES investors. Third, we identify design principles for Cap-and-Floor mechanisms that balance investment incentives, consumer protection, and overall system efficiency.

Methodologically, the paper extends the stochastic equilibrium framework of Mays et al. \cite{mays_asymmetric_2019} by endogenizing Cap-and-Floor contract pricing and incorporating alternative institutional arrangements for contract allocation. More broadly, it provides one of the first equilibrium assessments of revenue-stabilization mechanisms for LDES, offering policy insights into the effectiveness and design of Cap-and-Floor contracts in electricity markets characterized by investment risk and incomplete risk markets.

The remainder of the paper is organized as follows. Section~\ref{SEC: Methodlogy} presents the modeling framework and solution approach; Section~\ref{sec:StorageContracts} defines the LDES contract designs; Section~\ref{sec: solution} describes the solution method; Section~\ref{SEC:CaseStudySetup} presents the Great Britain case study; Section~\ref{SEC:Results} reports results; and Section~\ref{SEC:Conclusions} summarizes the conclusions and suggests areas for further research.

\section{Modelling methodology} \label{SEC: Methodlogy}

In this section, we present a two-stage stochastic investment equilibrium model and its solution method. This model is an adapted version of the framework in \cite{mays_asymmetric_2019},  further applied and extended, among others, in \cite{mays_financial_2023, shu_beyond_2023}. 

The model represents interactions among agents \( a \in \mathcal{A} \), with consumers \( c \in \mathcal{C} \), generators \( g \in \mathcal{G} \), and storage units \( s \in \mathcal{S} \). The latter two are collectively denoted as resources \( r \in \mathcal{R} = \mathcal{G} \cup \mathcal{S} \). The model also includes a set of hedging contracts \( \mathcal{K} \) that agents can buy and sell.  Uncertainty is modeled by a set of scenarios \( \Omega \), where each scenario \( \omega \in \Omega \) is associated with a probability \( p_\omega \), and describes realizations of exogenous inputs in the second stage. 

In the first stage, each investor determines investments and the volumes of hedging contracts to maximize its risk-adjusted utility, defined by a coherent risk measure $\rho_a$. Following standard practice, we use a convex combination of the expected value and CVaR of scenario-specific (production or consumption) surpluses \cite{rockafellar_optimization_2000}. In the second stage, a model that maximizes market surplus is solved to represent the competitive dispatch outcome and determine electricity prices and dispatch levels for each scenario \( \omega \). Second-stage results are also used to determine the scenario-specific contract payouts. 

\subsection{Second-stage problem} \label{sec:Dispatch}

The model $\text{DIS}_\omega$ in \eqref{eq:DIS} maximizes short-run market surplus $U_\omega$ by choosing the dispatch and consumption variables $\Gamma^{D}=\{d^{\text{fix}}_{\omega t},d^{\text{flex}}_{\omega t},q_{\omega t g},q^{\text{ch}}_{\omega t s},q^{\text{dis}}_{\omega t s},e_{\omega t s}\}$. The objective \eqref{eq:DIS_obj} aggregates (weighted by segment length $l_t$) consumer utility net of variable generation costs $C_g^V$, with period utility $b_{\omega t}$ defined by the willingness-to-pay function in \eqref{eq:DIS_utility}. Demand comprises a fixed segment $d^{\text{fix}}_{\omega t}$ up to $D^{\text{fix}}_{\omega t}$ valued at $B$ and a flexible segment $d^{\text{flex}}_{\omega t}$ up to $D^{\text{flex}}_{\omega t}$ with linearly decreasing marginal value. Feasibility is enforced by the power balance constraint \eqref{eq:DIS_balance}, generation ($q_{\omega t g}$) and storage ($q^{\text{ch}}_{\omega t s}$ and $q^{\text{dis}}_{\omega t s}$) power capacity limits \eqref{eq:DIS_gen}--\eqref{eq:DIS_power}, and storage state of charge $e_{\omega t s}$ dynamics \eqref{eq:DIS_soc} with energy bounds \eqref{eq:DIS_energy} enforcing energy-to-power ratio $\theta_s$. This is a convex problem, where dual variable $\lambda_{\omega t}$ represents the competitive wholesale electricity price \cite{steven_a_gabriel_complementarity_2013}.
\begin{subequations}
\label{eq:DIS}
\begin{gather}
\text{DIS}_\omega:\notag\\
\max_{\Gamma^{D}} \;
U_\omega
=
\sum_{t\in\mathcal{T}} l_t
\left(
b_{\omega t}
-
\sum_{g\in\mathcal{G}} C_g^V q_{\omega t g}
\right)
\label{eq:DIS_obj}\\
b_{\omega t}
=
B\!\left(
d^{\text{fix}}_{\omega t}
+
d^{\text{flex}}_{\omega t}
-
\frac{(d^{\text{flex}}_{\omega t})^2}{2D^{\text{flex}}_{\omega t}}
\right)
\quad \forall t
\label{eq:DIS_utility}\\
\begin{split}
\sum_{g\in\mathcal{G}} q_{\omega t g}
+
\sum_{s\in\mathcal{S}}
( q^{\text{dis}}_{\omega t s}-q^{\text{ch}}_{\omega t s})
&=
d^{\text{fix}}_{\omega t}
+
d^{\text{flex}}_{\omega t}
\\
&\forall \omega,t
\hspace{0.3cm}(\lambda_{\omega t})
\end{split}
\label{eq:DIS_balance}\\
0 \le d^{\text{fix}}_{\omega t} \le D^{\text{fix}}_{\omega t},
\quad
0 \le d^{\text{flex}}_{\omega t} \le D^{\text{flex}}_{\omega t}
\quad \forall \omega t
\label{eq:DIS_demand}\\
0 \le q_{\omega t g} \le c_g A_{\omega t g}
\quad \forall t,g
\hspace{0.3cm}(\mu_{\omega t g})
\label{eq:DIS_gen}\\
0 \le q^{\text{ch}}_{\omega t s} \le c_s,
\quad
0 \le q^{\text{dis}}_{\omega t s} \le c_s
\quad \forall t,s
\hspace{0.2cm}
(\upsilon^{\text{ch}}_{\omega t s},\upsilon^{\text{dis}}_{\omega t s})
\label{eq:DIS_power}\\
e_{\omega t s}
=
e_{\omega,t-1,s}
+
l_t
\left(
\eta^{\text{ch}}_s q^{\text{ch}}_{\omega t s}
-
\frac{q^{\text{dis}}_{\omega t s}}{\eta^{\text{dis}}_s}
\right)
\quad \forall t,s
\label{eq:DIS_soc}\\
0 \le e_{\omega t s} \le c_s \theta_s
\quad \forall t,s
\hspace{0.3cm}(\xi_{\omega t s}).
\label{eq:DIS_energy}
\end{gather}
\end{subequations}

\subsection{First-stage problem for investors} \label{sec:resources}

In the first stage, each resource $r\in\mathcal{R}$ determines investment $c_r$ and contracting $\nu_{rk}$ decisions by solving the problem $\text{INV}_r$ in \eqref{eq:GENg}. The corresponding first-stage decision variables are collected in the set
$\Gamma^{I}=\{c_r,\nu_{rk},\zeta_r,u_{\omega r},u^{+}_{\omega r}\}$, where all variables in $\Gamma^{I}$ are subject to the bounds specified in \eqref{eq:Investor4}--\eqref{eq:Investor5}. Each resource maximizes a risk-adjusted surplus \eqref{eq:Investor1}, defined as a convex combination of the expected value and the CVaR. Risk preferences are governed by parameter $\delta_r\in[0,1]$, which weights the expected value, and by $\psi_r$, which determines the tail probability at which CVaR is evaluated.

Equation \eqref{eq:Investor2} defines the scenario-specific surplus as a function of market net revenues, investment costs, and contract positions, where contract revenues depend on volumes $\nu_{rk}$, prices $\phi_k$, and payouts $\kappa_{\omega k}$. Since revenues and costs enter linearly in installed capacity, the value $\rho_r$ coincides with the marginal risk-adjusted profit of capacity expansion. Equation \eqref{eq:Investor3} implements the standard linear formulation of CVaR by introducing the Value-at-Risk (VaR) variable $\zeta_r$ and auxiliary variables $u^{+}_{\omega r}$, while \eqref{eq:Investor4} bounds contract volumes.
\begin{subequations}
\label{eq:GENg}
\begin{gather}
\text{INV}_r: \notag\\
\max_{\Gamma^{I}} \
\rho_r = (1 - \delta_r) \left( \zeta_r - \frac{1}{\psi_r} \sum_{\omega} p_\omega u^{+}_{\omega r} \right) \notag\\
+ \delta_r \sum_{\omega} p_\omega u_{\omega r}   \label{eq:Investor1} \\
u_{\omega r} = \pi_{\omega r} c_r - C_r^{\text{I}}c_r -\sum_k\nu_{r k}(\phi_k-\kappa_{\omega k}) \quad \forall \omega \label{eq:Investor2}\\
\zeta_r - u_{\omega r} \leq u^{+}_{\omega r}  \quad \forall \omega \label{eq:Investor3}\\ 
\underline{\nu}_{rk}\leq \nu_{r k} \leq \overline{\nu}_{rk} \quad \forall k \label{eq:Investor4}\\
0 \leq u^{+}_{\omega r} \quad \forall \omega. \label{eq:Investor5}
\end{gather}
\end{subequations}

Scenario-specific net revenues $\pi_{\omega r}$ are obtained from the dual variables of the second-stage dispatch problem. Using the duality of the market-clearing model, normalized profits of generators and storage resources are represented as functions of scarcity rents \cite{ehrenmann_generation_2011}, and are defined as follows:
\begin{gather}
\pi_{g \omega} = \sum_{t} \mu_{\omega t g} A_{\omega t g} \quad \forall g,  \omega \\
\pi_{s \omega} = \sum_{t}\left( \upsilon_{\omega t s}^{dis} + \upsilon_{\omega t s}^{ch} + \xi_{\omega t s}\theta_s \right) \quad \forall s, \omega
\end{gather}

\subsection{First-stage problem for consumers}\label{sec: consumers}

In the first stage, consumers decide the volume of hedging contracts to support by solving the problem $\text{CON}_c$ in \eqref{eq:CON}. Assuming perfect competition on the demand side, we model a single representative consumer. The first-stage decision variables of the consumer are collected in the set
$\Gamma^{C}=\{\nu_{ck},\zeta_c,u_{\omega c},u^{+}_{\omega c}\}$, where all variables in $\Gamma^{C}$ are subject to the bounds specified in \eqref{eq:Consumer4}--\eqref{eq:Consumer5}. The consumer maximizes a risk-adjusted measure of surplus \eqref{eq:Consumer1}, defined as a convex combination of the expected value and the CVaR of scenario-specific surplus outcomes. Risk preferences are governed by parameters $\delta_c$ and $\psi_c$, which weight the expected value and determine the CVaR tail probability, respectively.

Scenario-specific consumer surplus is defined in \eqref{eq:Consumer2} as consumption utility net of electricity procurement costs and net contract payments, where contract revenues depend on volumes $\nu_{ck}$, prices $\phi_k$, and payouts $\kappa_{\omega k}$. Equation \eqref{eq:Consumer3} implements the standard linear formulation of CVaR by introducing the VaR variable $\zeta_c$ and the auxiliary variable $u^{+}_{\omega c}$.
\begin{subequations}
\label{eq:CON}
\begin{gather}
\text{CON}_c: \notag\\
\max_{\Gamma^{C}} \
\rho_c = (1 - \delta_c) \left( \zeta_c - \frac{1}{\psi_c} \sum_{\omega} p_\omega u^{+}_{\omega c} \right) \notag\\
+ \delta_c \sum_{\omega} p_\omega u_{\omega c}  \label{eq:Consumer1} \\
u_{\omega c} = \sum_{t} l_{t} (b_{\omega t}
- \lambda_{\omega t}\left(d_{\omega  t}^{\text{fix}} + d_{\omega  t}^{\text{flex}}\right) \notag\\ 
+ \sum_k\nu_{c k}(\phi_k-\kappa_{\omega k}) \quad \forall \omega \label{eq:Consumer2}\\ 
\zeta_c - u_{\omega c} \leq u^{+}_{\omega c}  \quad \forall \omega \label{eq:Consumer3}\\ 
\underline{\nu}_{ck}\leq \nu_{c k} \leq \overline{\nu}_{ck} \quad \forall k \label{eq:Consumer4}\\
0 \leq u^{+}_{\omega c} \quad \forall \omega. \label{eq:Consumer5}
\end{gather}
\end{subequations}

Total annual demand differs across scenarios because scenario-specific demand profiles determine achievable consumer surplus, potentially biasing the distribution's lower tail toward low-demand scenarios. The consumer surplus achieved when all demand is satisfied varies across scenarios because the demand parameters (i.e., $D^{fix}_{\omega,t}$) differ across scenarios. To ensure a fair comparison of consumer surplus across scenarios, we adjust the consumer probabilities as described in \eqref{EQ: Normalization}. Define
\begin{equation} \label{EQ: Normalization}
\bar{D}_{\omega}^{\mathrm{fix}} := \sum_{t\in\mathcal{T}} l_t\, D_{\omega t}^{\mathrm{fix}},
\; 
\bar{D}^{\mathrm{fix}} := \sum_{\omega\in\Omega} p_\omega\, \bar{D}_{\omega}^{\mathrm{fix}} ,
\; 
p_\omega' := p_\omega \frac{\bar{D}_{\omega}^{\mathrm{fix}}}{\bar{D}^{\mathrm{fix}}}
\end{equation}
with $p_\omega'$ being the adjusted weight, so that tail-risk assessments are not driven by variation in total annual demand across scenarios. In our model $p_\omega'$ replaces $p_\omega$ in \eqref{eq:Consumer1}. 

Alternatively, following~\cite{mays_asymmetric_2019}, consumer surplus could be measured relative to a full-service benchmark, or more generally recentered so that the maximum achievable consumer surplus is the same across years.

\subsection{Contractual regimes}\label{sec: Contractual_regimes}

We consider two contractual regimes. The payout functions are exogenously specified in both cases. We start with \emph{zero-premium} contracts, which might be imposed by a regulator seeking to avoid the welfare losses associated with incomplete risk markets. This regime captures policy intervention of a regulator (or equivalent state agency) motivated by the welfare losses associated with incomplete risk markets and suppressed investment, in the spirit of Laffont and Tirole \cite{laffont_using_1986}. However, depending on the payout function, a zero-premium contract might involve significant expected transfers between agents. In commercial settings, agents normally require an up-front premium in return for accepting such an "out of the money" contract, and so our second \emph{negotiated} regime allows this premium $\phi_k$ to be determined endogenously through competitive interaction between investors and consumers.

\section{LDES-Dedicated Contracts} \label{sec:StorageContracts}
We focus on C\&F contracts under alternative parametrizations proposed in Great Britain \cite{ofgem_long_2025}. However, our methodology can investigate a range of alternative contracts proposed in the academic and institutional literature, such as spread CfDs \cite{mastropietro_taxonomy_2024}, revenue put options in Australia \cite{nsw_department_of_planning_industry_and_environment_long-term_2021}, or the Index Storage Credit Mechanism proposed in New York \cite{new_york_state_energy_research_and_development_authority_over_2024}. 

The scenario-specific payouts of contracts, denoted by \(\kappa_{\omega k}\), are calculated based on the second-stage dispatch outcomes. When \(\kappa_{\omega k}\) is positive, money is transferred from consumers to investors, and in the reverse direction when \(\kappa_{\omega k}\) is negative. Each contract type is characterized by a payout function that takes as inputs the outputs of the second-stage dispatch problem, \(\text{DIS}_\omega\), and provides \(\kappa_{\omega k}\) in value per MW signed.

C\&F is a mechanism that provides a multi-year contract during which the project's net revenues are subject to predefined caps and floors. The floor guarantees a minimum revenue for investors, thereby encouraging investment in qualifying assets. Conversely, the cap allocates excess returns to consumers, safeguarding them from elevated electricity bills. We formulate it after the scheme rolled out in the UK \cite{ofgem_long_2025}.  In a general form, with normalized (i.e., per MW installed) cap and floor defined as $\overline{\pi}_r$ and $\underline{\pi}_r$ respectively, the payout is calculated as

\begin{gather}
    \kappa_{\omega,r} = (1-\alpha^F) \cdot \max(0,\underline{\pi}_r - \pi_{\omega,r}) \notag\\ - (1 - \alpha^C) \cdot \max(0, \pi_{\omega,r} - \overline{\pi}_r).
    \label{EQ: CapAndFloorPayoff}
\end{gather}
where $\alpha^C,\alpha^F  \in [0,1]$ are factors determining how returns short of the floor and in excess of the cap will be shared between consumers and LDES. When these factors are 0, the cap and floor are called hard. Conversely, when their value is non-zero, the cap or floor is called soft. 

Choosing parameter values for the C\&F is challenging. The regulator's goal in setting cap and floor levels is to ensure investability for LDES by guaranteeing cost recovery with a return at the floor, without creating undue risks for consumers and without distorting dispatch incentives. The desire to preserve dispatch incentives explains non-zero values for the sharing factors ($\alpha^C,\alpha^F$).

\subsubsection{Determining the floor and cap levels}
Both the floor and the cap can be set with reference to technology-specific capital costs. Let $\text{CRF}(\cdot)$ denote the capital recovery factor, which annualizes investment costs given a discount rate and lifetime. For an asset $r$ with specific investment cost $C_r^I$ and lifetime $l_r$, a generic cap/floor level can be expressed as
\begin{equation} \label{eq: FloorLevel}
\pi_r^{\text{level}} = C_r^I \cdot \text{CRF}(\iota_r, l_r),
\end{equation}
where $\iota_r$ is the rate used to compute the annualized cost (e.g., a WACC-based or policy-defined rate of return). Under risk neutrality, the long-run average net revenue would cover $C_r^I \cdot \text{CRF}(\text{WACC}_r,l_r)$; therefore, plausible floor levels should be below this benchmark and plausible cap levels above it.\footnote{UK regulatory discussion and supporting analysis \cite{cepa_cap_2025} highlighted three floor specifications that provide identical protection to debt service but differ in the degree of equity protection. \emph{ (Debt service only):} guarantees debt service but provides no explicit protection to equity. \emph{ (Debt service + equity capital recovery):} guarantees debt service and recovery of equity principal. \emph{ (Debt-like return on total capital):} treats equity analogously to debt by applying the same recovery logic to total capital. In Section \ref{SEC:CaseStudySetup}, we highlight what these levels correspond to in terms of values in \eqref{eq: FloorLevel}, and in Section \ref{SEC:Results}, we run a range of floors that encompass those. 
}

\section{Solution method} \label{sec: solution}

Problems $\text{DIS}_\omega$, $\text{INV}_r$, and $\text{CON}_c$, defined for all scenarios $\omega \in \Omega$, resources $r \in \mathcal{R}$, and consumers $c \in \mathcal{C}$, jointly define the equilibrium problem $\text{EQ}$. Formally, an equilibrium is characterized by the following two key conditions:
\begin{gather}
0\leq c_r \perp -\rho_r \geq 0, \quad \forall r \in \mathcal{R} \label{EQ:InvestmentCondition}\\
\sum_{a\in \mathcal{A}}\nu_{ak}=0, \quad \forall k\in\mathcal{K} \label{EQ:ContractCondition}
\end{gather}

Equation \eqref{EQ:InvestmentCondition} states that investment in the specific capacity will continue until risk-adjusted profit is zero \cite{ehrenmann_generation_2011}. Equation \eqref{EQ:ContractCondition} defines the balance of all contractual trades for all participants and contract types \cite{shu_beyond_2023}. 

When individual agents’ problems are convex and satisfy constraint qualifications, the equilibrium system can be reformulated as a mixed-complementarity problem (MCP) via the Karush–Kuhn–Tucker (KKT) conditions \cite{steven_a_gabriel_complementarity_2013} and solved using Newton-type solvers such as PATH. However, this approach is computationally tractable only for very small instances \cite{hoschle_admm-based_2018}. Under certain conditions \cite{steven_a_gabriel_complementarity_2013}, these problems can also be formulated as an equivalent centralized optimization model and easily solved by off-the-shelf solvers. In the case of incomplete risk-averse equilibria, the straightforward centralized convex problem does not exist, due to the non-monotone mapping of players' strategies ~\cite{abada_multiplicity_2017}. As such, alternative methods are required to solve such problems at scale. 

\subsection{Benchmark algorithm} \label{sec: benchmark_algorithm}

The foundation of our work is the algorithmic equilibrium search method, a variant of the Gauss-Seidel and ADMM algorithms, developed in the original work by Mays et al. \cite{mays_asymmetric_2019}. The algorithm iterates between the first and second stages, gradually adjusting the invested capacities based on investors' risk measures, and settles contract prices based on volume imbalances. The algorithm terminates when the contracts are balanced (i.e., condition \eqref{EQ:ContractCondition} is met) and the risk-adjusted profits of all investors are near zero (i.e., condition \eqref{EQ:InvestmentCondition} is met). 

\subsection{Simplifications and solution method}

The algorithm in \cite{mays_asymmetric_2019} can accommodate multiple contract types and investors, as further demonstrated by \cite{shu_beyond_2023}. To simplify the equilibrium structure and reflect the institutional design of LDES support mechanisms, we impose two simplifying assumptions.

\textit{Assumption 1 (Single-technology contracting).}
The contract set $\mathcal{K}$ is a singleton, i.e., $|\mathcal{K}| = 1$, and contract trading is permitted only between the single investor and consumers.

\textit{Assumption 2 (Restricted contract volumes).}
Contract volumes are restricted to corner points, $\nu_{rk} \in \{0, c_r\}$, $\forall r \in \mathcal{R}, k \in \mathcal{K}$, where $c_r$ denotes the installed capacity of the resource under contract.

Under Assumption 1, investment decisions for all candidate technologies other than LDES depend only on the equilibrium LDES capacity and are independent of the contract design used to support that capacity. Based on this observation, we solve a stochastic risk-neutral capacity expansion problem over a discrete grid of LDES capacities to determine the corresponding equilibrium investments across all non-LDES technologies.
We then use the resulting energy prices to evaluate contract payoffs and risk measures. As validation, the original iterative equilibrium search from \cite{mays_asymmetric_2019} is applied to a subset of cases, yielding identical capacity outcomes.

In terms of Assumption 2, we note that interior optima cannot be ruled out under continuous $\nu \in [0, c_r]$, as tail reordering with increasing volume can cause the marginal value of contracting to change sign before reaching $c_r$. The assumption is, therefore, both an institutional simplification, reflecting that LDES support schemes are often structured as full hedges, and a computational one, as it enables a tractable grid search over a wide range of contracts and risk parameters. Its empirical validity is confirmed by running the algorithm of \cite{mays_asymmetric_2019} with $\nu$ unconstrained in the range $[0, c_r]$, where we did not observe any intermediate optimal volumes.

\subsection{Mutually acceptable contract premiums}

For a given capacity mix, under Assumptions 1 and 2, acceptable premiums for both sides can be analytically defined for a given capacity mix.\footnote{Another interpretation of this derivation may refer to rational post-investment decisions when the investment costs are sunk, and impacts of contracts on the capacity available in the market are accounted for. In this case, the acceptable premiums depend on the contract payouts, under a risk-adjusted probability measure that considers the overall surplus/profit.} Define $\rho_a^{0}$ as the agent's risk measure without a contract, and $\rho_a^{*}$ after adjusting to contract payouts at full contracted capacity (but no up-front premium). Since the risk measure $\rho$ is assumed to be coherent, it has a translation-invariant property, and with a constant premium across scenarios $\Omega$, its risk changes by that exact amount. Then investors require:
\begin{gather}
\rho_r^{0} \leq  \rho_r^{*} - \nu_k\phi_k \Rightarrow \Bar{\phi}_{rk} = \frac{\rho_r^{*} - \rho_r^{0}}{\nu_k} \quad \forall k\in\mathcal{K},
\label{EQ:PriceThesholdsInvestor}
\end{gather}
where $\Bar{\phi}_{rk}$ is the maximum premium investor $r$ is willing to pay. Similarly, for consumers, yet with the opposite sign: 
\begin{gather}
\rho_c^{0} \leq  \rho_c^{*} + \nu_k\phi_k \Rightarrow \underline{\phi}_{ck} = \frac{\rho_c^{0} - \rho_c^{*}}{\nu_k} \quad \forall k\in\mathcal{K},
\label{EQ:PriceThesholdsConsumer}
\end{gather}
where $\underline{\phi}_{ck}$ is a minimum premium consumer $c$ is willing to accept. For the range of acceptable premiums to exist, we require $\underline{\phi}_{ck}\leq\Bar{\phi}_{rk}$. Negative premiums are possible if payouts are skewed towards consumer benefit and consumers are willing to pay for the contract. If this inequality does not hold, we assume that there is no agreeable premium and the contract cannot be signed ($\nu_{ak}=0$).

Under Assumptions 1 and 2, \eqref{EQ:PriceThesholdsInvestor}\textendash \eqref{EQ:PriceThesholdsConsumer} are analytical bounds for the range of contract premiums found by the inner loop (ADMM-like) of the Algorithm in \cite{mays_asymmetric_2019}. Therefore, multiple values for contract premiums are possible. In practice, the level of contract premium will be negotiated and 
will primarily depend on the agents' bargaining power. For most of our results, we choose a premium at the center (i.e., $(\underline{\phi}_{ck} +\Bar{\phi}_{rk})/2$) of the range of mutually acceptable contract premiums (i.e. $[\underline{\phi}_{ck},\Bar{\phi}_{rk}$]). We show in Appendix 1 that this solution corresponds to the Nash Bargaining solution \cite{nash_bargaining_1950}. 

\section{Case study} \label{SEC:CaseStudySetup}

The case study represents a simplified GB system in 2035. Existing unabated gas and nuclear are set to 35 GW and 5 GW, respectively. The model invests in renewable and storage technologies. Storage options include a 2-hour battery and a 12-hour LDES. Renewable options include solar PV, onshore wind, and offshore wind. Onshore wind investment is limited to 30 GW. The peak demand is 65 GW according to the NESO's FES projections \cite{eso_future_2024}. The price cap is assumed to be 20,000 \$/MWh, and the flexible demand is 2 GW in all time steps. Set $\Omega$ comprises 60 equiprobable scenarios that combine 20 weather scenarios with three scenarios for short-run variable costs of gas power plants (78, 150, 220) \$/MWh, reflecting both fuel price and carbon cost variability. The model horizon is a year with 1440 time steps of varying duration per scenario, obtained using the technique of \cite{pineda_chronological_2018}. 

We model investments in LDES by a risk-averse investor who computes the CVaR at a tail probability of $\Psi=0.2$. To assess the impact of risk aversion on our results, we vary the weight the investor assigns to the expected value ($\delta_r$). For batteries and renewables, we assume that investors are risk-neutral ($\delta_r =1$). Because both technologies have been deployed at scale in the UK, we assume that existing financial markets and government support mechanisms adequately de-risk investments in these technologies, which may also access revenue streams from ancillary services not explicitly modeled here. Consumers are assumed to be risk-averse, with equal weight assigned to the expected value and CVaR.

Contracts facilitate risk sharing between the risk-averse LDES investor and consumers through the Cap-and-Floor (C\&F) mechanism described in Section~\ref{sec:StorageContracts}. In the case study, we evaluate a range of parametrizations by varying the cap rate of return \(i_r\) from 10\% to 15\% and the floor rate of return from 2\% to 7\%. For each cap--floor pair, we consider both hard designs (no sharing) and soft designs with 10\% sharing, i.e., \(\alpha^C=\alpha^F=0.10\). Techno-economic assumptions about technologies are presented in Table \ref{tab:table1}.

\begin{table}[!htbp]
\centering
\caption{Techno-economic Parameters of Considered Technologies}
\label{tab:table1}
\setlength{\tabcolsep}{2pt}
\scriptsize
\begin{tabular}{lccccc}
\hline
\textbf{Parameter} & \rotatebox{0}{\parbox{1cm}{\centering Solar}} & \rotatebox{0}{\parbox{1cm}{\centering Onshore\\Wind}} & \rotatebox{0}{\parbox{1cm}{\centering Offshore\\Wind}} & \rotatebox{0}{\parbox{1cm}{\centering Battery (2h)}} & \rotatebox{0}{\parbox{1cm}{\centering LDES (12h)}} \\ \hline
Capex (\$/kW) & 1000 & 1400 & 3500 & 600 & 4500 \\
FOM (\$/kW-yr) & 15 & 28 & 50 & 33 & 20.2 \\
Charge Eff. (\%) & - & - & - & 92 & 80 \\
Discharge Eff. (\%) & - & - & - & 92 & 80 \\
Lifetime (years) & 25 & 30 & 30 & 25 & 50 \\
WACC (\%) & 6.2 & 6.2 & 7.7 & 6.5 & 7.1 \\
\hline
\end{tabular}

\vspace{0.3em}
\noindent\scriptsize{Note: For all technologies, the debt share is assumed to be 70\%, and the tax rate is assumed to be 25\%. The weighted average cost of capital (WACC) reflects the blended cost of debt and equity financing. The interest rate (for LDES assumed 8\%) corresponds to the pre-tax cost of debt, while the return on equity reflects the required return to equity investors (for LDES assumed 9.8\%), which is typically higher due to their greater exposure to risk.}
\end{table}

A commonly used benchmark in incomplete-market settings is a complete market \cite{dimanchev_consequences_2024, mays_asymmetric_2019}. Under certain conditions (e.g., uniform (across all agents) risk parameters), the complete market problem can be reformulated as a single optimization problem \cite{ralph_risk_2015}. In our case, some agents are risk-neutral and others risk-averse, and the probabilities used by consumers are slightly different than those used by investors (due to \eqref{EQ: Normalization}). Following Corollary 3 and Example 3.1 of Ralph and Smeers \cite{ralph_pricing_2011}, we note that in our case, the risk set of the risk-neutral agent is a subset of the risk set of other agents. Consequently, the complete-market solution collapses to the solution of a single risk-neutral capacity expansion problem with probabilities $p_\omega$, which we subsequently use as a benchmark. 

Investment performance is evaluated using the internal rate of return (IRR), defined as the discount rate satisfying $\sum_{n=1}^{L_r} CF_r/(1+\text{IRR}_r)^n = C_r^I$ for a constant annual cash flow $CF_r$ and lifetime $L_r$. Scenario-specific values $\text{IRR}_{r\omega}$ are computed from second-stage net revenues $\pi_{r\omega}$, and average performance is summarized by $\text{IRR}_r(\mathbb{E}[CF_r(\omega)], L_r, C_r^I)$, used to evaluate the implied risk premium relative to a benchmark risk-free rate \cite{ehrenmann_generation_2011, mays_financial_2023}.

\section{Results} \label{SEC:Results}

\subsection{Investments under fully incomplete markets} \label{sec: Res1}

Fig. \ref{FIG:RN_complete_comparison} presents the risk-neutral mix and the change in investment as the risk-aversion parameter $\delta$ for the LDES investor is varied. The risk-neutral equilibrium serves as the efficient benchmark against which the effects of market incompleteness are evaluated. The risk-neutral mix has both LDES and batteries, with the former's power capacity being higher.\footnote{We note that the clustering method we use \cite{pineda_chronological_2018}, while suitable for modeling LDES, can smooth daily variability and, as such, underestimate the value of batteries and overestimate the value of solar power \cite{mannhardt_accurately_2025}.} As $\delta$ decreases (and the investor is more risk-averse), the installed capacity of LDES declines, eventually disappearing from the mix for $\delta<0.55$, with batteries and new solar capacity filling the gap.\footnote{These changes in the mix may affect the business models of other investors, represented as risk-neutral in this study. Evaluation of shifts in these business models is beyond the scope of this work.} In both cases, the substitution is also associated with a change in the system's total installed energy storage capacity. For the subsequent analysis of the cap-and-floor mechanism in this article, we set $\delta = 0.6$ to reflect limited LDES investment in the absence of de-risking mechanisms.

\begin{figure}[!htbp]
    \centering
    \includegraphics[width=\columnwidth]{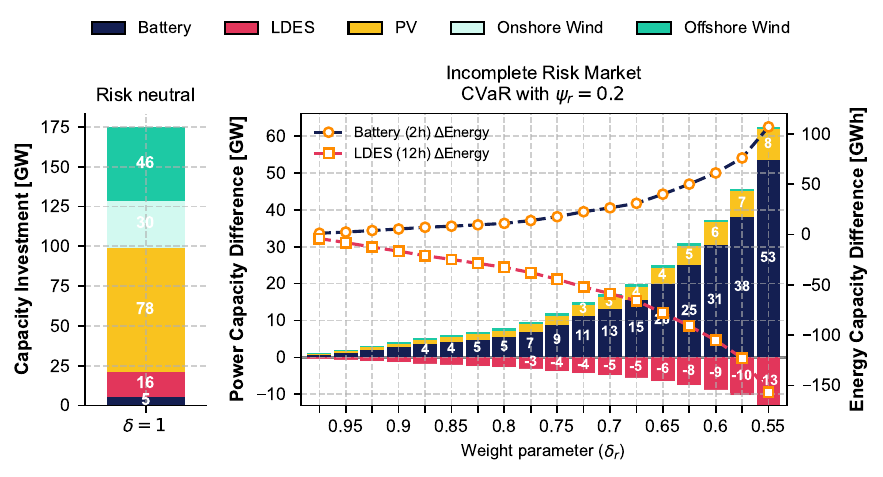}
    \caption{(left): Equilibrium capacity mix by technology under risk neutrality. (right): Capacity deviations from the risk-neutral benchmark as LDES risk aversion increases ($\delta$ from 1 to 0.55); positive values indicate higher capacity than the benchmark, negative values indicate lower. Bars denote power capacity; lines denote storage energy capacity.}
    \label{FIG:RN_complete_comparison}
\end{figure}

The observed substitution of LDES for batteries as risk aversion increases suggests that LDES policies or dedicated mechanisms will have secondary effects on storage choices. In the absence of policies that de-risk LDES investments, battery investments might exceed the risk-neutral benchmark, while excessive support for LDES could push them below it.

\subsection{LDES return distributions}

The left panel of Fig. \ref{FIG:LDES_net_revenue_cdf_adjusted} shows the distributions of scenario-specific IRRs for LDES with changing risk aversion. At the risk-neutral equilibrium ($\delta=1$), the IRR distribution is clearly skewed towards negative values, and the expected IRR equals the risk-free WACC (7.1\%).  A risk-averse investor can manage risk either by adjusting its investment or hedging its exposure through contracts. When it can only do the former, it reduces its investment, thereby increasing prices and raising returns in most scenarios to compensate for the weight placed on tail losses. The more pronounced the loss tail, the greater the potential impact of risk aversion on the cost of capital. In this case, the average IRR is 9.8\% under $\delta=0.6$, which is almost 3\% higher than the risk-free WACC. This increase in the implicit cost of capital makes investment in LDES less attractive to consumers, shifting the competitive mix towards other technologies that would appear less cost-effective if all investors were risk-neutral.

\begin{figure}[!htbp]
    \centering
    \includegraphics[width=\columnwidth]{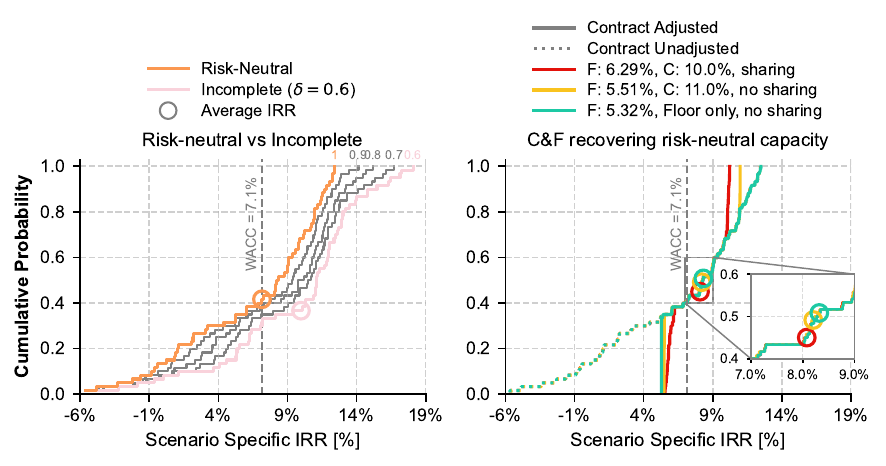}
    \caption{Empirical CDFs of scenario-specific IRRs for LDES: (left) under varying levels of risk aversion and (right) for selected C\&F parametrizations at $\delta = 0.6$ that support risk-neutral capacity under zero-premium contracts. Colors distinguish contract types and benchmarks; line styles distinguish contract-adjusted from unadjusted distributions.}
    \label{FIG:LDES_net_revenue_cdf_adjusted}
\end{figure}

Introducing a revenue floor truncates the lower tail of the distribution, reducing downside risk (increasing CVaR) and allowing the investor to accept a lower average return. The opposite is true when the cap binds: it limits upside potential, shifting the distribution to the left and discouraging investment. In general, high floors need to be paired with low caps to achieve a target level of investment (see right panel of Fig. \ref{FIG:LDES_net_revenue_cdf_adjusted}). Considering that the return distribution is negatively skewed and asymmetric, symmetric cap and floor levels around the risk-free WACC are generally not optimal.

For the same level of invested capacity, a tighter collar reduces the average IRR (zoom-in box in Fig. \ref{FIG:LDES_net_revenue_cdf_adjusted}). This de-risking (reflected as a reduction in average IRR) is one of the goals of C\&F schemes. However, the tighter the collar, the more frequently the cap and floor are binding, potentially weakening the LDES's incentive to respond to wholesale market signals. 

The feedback between contract design, market equilibrium, and investor returns, illustrated in Fig. \ref{FIG:LDES_net_revenue_cdf_adjusted}, cannot be captured by static financial assessments that treat prices or revenues as exogenous. The availability of contracts influences investment decisions, which in turn reshape electricity prices and the distribution of project revenues. So,  LDES investors considering the C\&F scheme must also project their revenue distributions. That's why clear and credible information about future market design and contract rules is important if Cap-and-Floor mechanisms are to reduce financing costs effectively.

\subsection{Impact of cap-and-floor on LDES investment}

The introduction of administered, zero-premium C\&F contracts substantially increases equilibrium investment in LDES by reducing investors' exposure to downside revenue risk. As shown in the left panel of Fig.~\ref{FIG:capacity_risk_zeroCF}, LDES capacity increases monotonically with the floor rate of return. A higher floor reduces the likelihood and severity of low-revenue outcomes, lowering the compensation investors require for bearing risk. Consequently, the required risk premium declines as the floor increases (right panel of Fig.~\ref{FIG:capacity_risk_zeroCF}), stimulating additional investment. The risk premium falls from approximately 2.8\% (incomplete market) to almost zero when the floor approaches the risk-free WACC. At this point, the investor is effectively fully insured against downside risk and invests up to the maximum capacity permitted by the C\&F mechanism.

\begin{figure}[!htbp]
\centering
\includegraphics[width=\columnwidth]{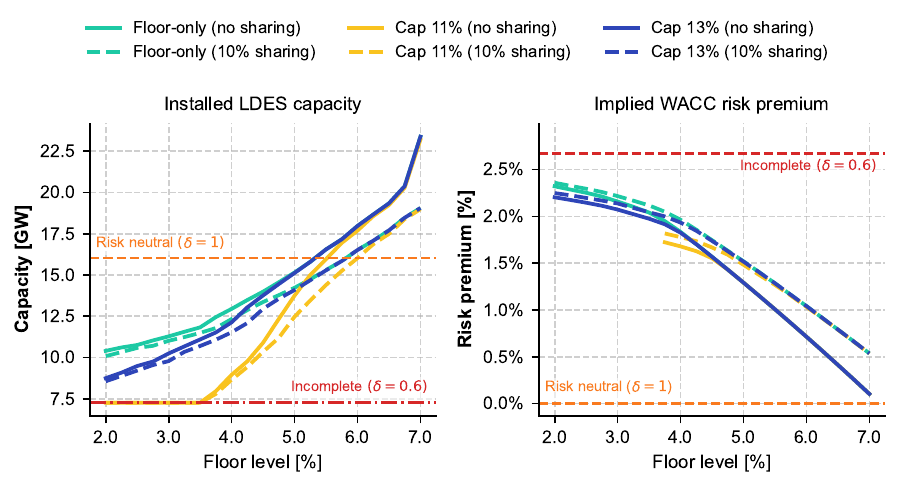}
\caption{(left) LDES capacity and (right) risk premium versus floor rate for cap-and-floor configurations at $\delta = 0.6$.}
\label{FIG:capacity_risk_zeroCF}
\end{figure}

The effect of the cap is qualitatively different. For a given floor, a tighter cap generally reduces equilibrium LDES investment by limiting the project's upside potential and the rate of return the investor can obtain. This creates a trade-off between investment incentives and risk reduction. For example, at a 4\% floor, reducing the cap from 13\% to 11\% results in lower equilibrium investment despite a lower risk premium. This behavior highlights the importance of equilibrium modeling. A central-planner approach that assumes the cost of capital exogenously would likely predict higher investment whenever financing costs fall, whereas in the stochastic equilibrium model, the impact of the contract on the revenue distribution is endogenously captured.

The interaction between the floor and the cap is also noteworthy. As the floor increases, the influence of the cap diminishes, as indicated by the convergence of the investment curves in Fig.~\ref{FIG:capacity_risk_zeroCF}. A higher floor shifts the revenue distribution, making the cap bind less frequently and therefore reducing its effect on both investment and risk. Conversely, combinations of a low floor and a tight cap may fail to provide sufficient compensation for downside risk, making participation unattractive and reducing investment to levels close to those observed under incomplete markets. In this case study, this result is observed for floors below 3.5\% and an 11\% cap. 

Introducing a revenue-sharing factor weakens the de-risking effect of the C\&F mechanism. Relative to hard contracts, soft contracts with revenue sharing leave investors partially exposed to revenue fluctuations outside the cap and floor, increasing the required risk premium and reducing equilibrium investment. The magnitude of this effect depends on the contract parameterization: it is most pronounced when either the cap or the floor binds frequently and becomes relatively small when both limits are rarely active.

Overall, the results demonstrate that the floor, cap, and sharing parameters cannot be considered independently. Multiple combinations of floor, cap, and sharing parameters can deliver similar equilibrium investment while implying different levels of investor risk and different transfers between consumers and LDES investors. Restoring the risk-neutral level of investment, therefore, requires progressively higher floor levels as either the sharing factor increases or the cap becomes more restrictive. These strong interactions suggest that Cap-and-Floor contracts should be calibrated as an integrated package rather than through sequential adjustment of individual parameters.

\subsection{Impact of cap-and-floor on system and consumers}
Cap-and-Floor (C\&F) contracts redistribute revenues between consumers and LDES investors. When market revenues fall below the floor, consumers compensate investors for the shortfall, whereas revenues above the cap are returned by investors to consumers.

Fig.~\ref{FIG:horizontal_zeroCF} presents contract payouts and changes in social and consumer surplus. At the contract parameterizations that restore the risk-neutral level of LDES investment, consumers make expected net transfers of $\sim$\$0.7--0.8B to LDES investors. These transfers are similar to the system-wide benefits of increased LDES deployment. Relative to the incomplete-market equilibrium, social welfare increases by $\sim$\$0.6B, while consumers capture substantially larger net benefits of $\sim$\$7.6-7.8B. The large consumer surplus effect reflects both lower system costs and a redistribution of rents from fixed-capacity technologies (nuclear and unabated gas).\footnote{Depending on the regulatory framework, energy revenues may also accrue to consumers, resulting in reduced improvements in consumer surplus due to changes in the energy revenues of fixed-capacity technologies.} 
These results indicate that consumers have the most to lose from incomplete markets and, in expectation, are better off funding C\&F contracts because the benefits of increased LDES investment exceed the expected contract payouts.

The introduction of C\&F contracts also affects competing investment opportunities. In equilibrium, all candidate technologies without a cap continue to earn their required rate of return. Technologies with capped capacity earn more than the required rate of return, but their markup reduces as LDES capacity increases. Overall, no candidate technology becomes unprofitable. Instead, contract-induced changes in investment volumes and scarcity rents can lead some technologies to earn lower absolute profits. 

Fig.~\ref{FIG:horizontal_zeroCF} further shows that the utility of risk-averse consumers increases monotonically with LDES deployment, even beyond the risk-neutral investment level. This reflects consumers' preference for outcomes with less severe downside welfare realizations: additional LDES improves the lower tail of the welfare distribution. By contrast, when welfare is evaluated by a risk-neutral scheme administrator, the additional gains from investment beyond the risk-neutral capacity are insufficient to compensate for the larger expected contract payouts. This divergence highlights that the preferred capacity level depends on whether the objective is to maximize expected welfare or risk-adjusted welfare, given risk-averse consumers.

\begin{figure*}[!htbp]
\centering
\includegraphics[width=2\columnwidth]{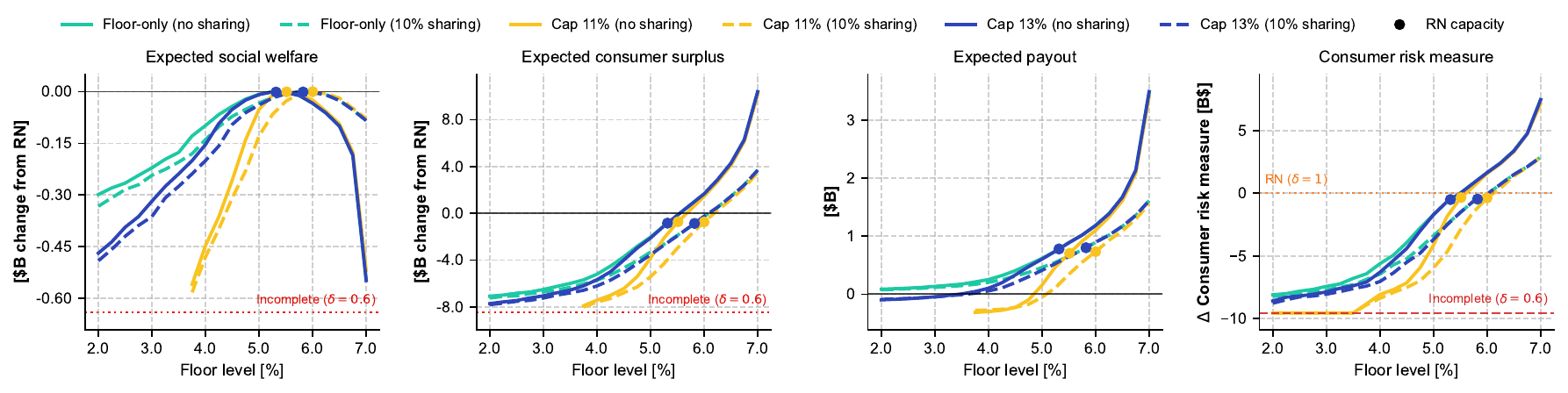}
\caption{{(left) change in the expected social welfare from the risk-neutral level; (second from left): change in the expected consumer welfare from the risk-neutral level; (second from right): expected total payouts. (right) risk-adjusted consumer welfare changes from the risk-neutral baseline. All figures are for selected LDES cap-and-floor designs at $\delta = 0.6$ and zero-premium C\&F. Markers indicate parametrizations recovering risk-neutral capacity.}}
\label{FIG:horizontal_zeroCF}
\end{figure*}

\subsection{Negotiated contracts} \label{sec: negotiated_con}
Under negotiated contracts, equilibrium LDES capacity continues to increase with the floor rate of return (left panel of Fig.~\ref{FIG:capacity_risk_negotiatedCF}), but the sensitivity of investment to the contract design is substantially weaker than under zero-premium contracts. Across most parameterizations, equilibrium capacity clusters around a similar level, indicating that negotiated premiums absorb much of the variation that cap, floor, and sharing parameters would otherwise introduce.

The relationship between the cap and investment is also fundamentally different from that observed under zero-premium contracts. Whereas tighter caps reduced investment when contracts were provided at zero contract premium, they now increase equilibrium LDES capacity. This reversal occurs because tighter caps increase the value of the insurance provided by the contract to consumers, thereby increasing the premium consumers are willing to pay (or, equivalently, lowering the premium they are willing to accept). The resulting transfer compensates investors for the reduction in upside revenues, allowing additional investment.

The right panel of Fig.~\ref{FIG:capacity_risk_negotiatedCF} shows that negotiated contracts continue to reduce investors' required risk premia as the floor increases, albeit to a much smaller extent than under zero-premium contracts. Unlike the zero-premium case, the equilibrium risk premium now depends strongly on the cap level, with tighter caps resulting in lower required returns because negotiated premiums can also compensate investors for revenue risk.

\begin{figure}[!htbp]
\centering
\includegraphics[width=\columnwidth]{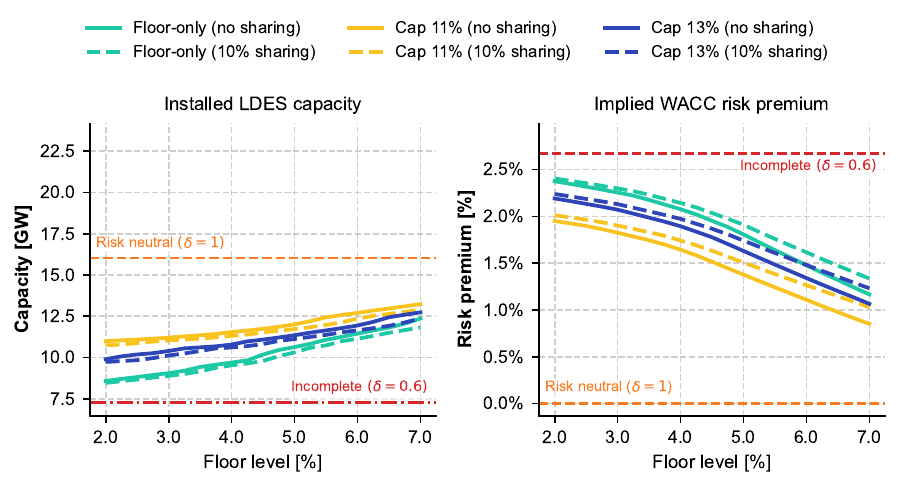}
\caption{(left) LDES capacity and (right) risk premium versus floor rate for cap-and-floor configurations at $\delta = 0.6$ for negotiated contracts.}
\label{FIG:capacity_risk_negotiatedCF}
\end{figure}

Despite reducing financing costs, negotiated contracts generally fail to restore LDES investment to its risk-neutral level. Fig. ~\ref{FIG:horizontal_negotiatedCF} illustrates why. In contrast to zero-premium contracts, negotiated contracts generate almost no expected transfers between consumers and investors. Consumers value the contracts only for the risk reduction they provide at the prevailing level of LDES capacity and negotiate premiums accordingly. They do not internalize that a more generous premium (or even a net expected transfer) would induce additional LDES investment, which would subsequently reduce electricity prices, improve resource adequacy, and increase consumer welfare. 

The resulting equilibrium is therefore characterized by a coordination failure. There exists a combination of contract premium and LDES investment that satisfies investors' zero risk-adjusted profit condition while simultaneously delivering higher consumer surplus than the negotiated outcome. However, because bargaining is conducted taking LDES capacity as given, this mutually beneficial equilibrium is never reached. The contract market prices financial risk but does not fully internalize the consumer benefits associated with additional LDES investment.

\begin{figure*}[!htbp]
\centering
\includegraphics[width=2\columnwidth]
{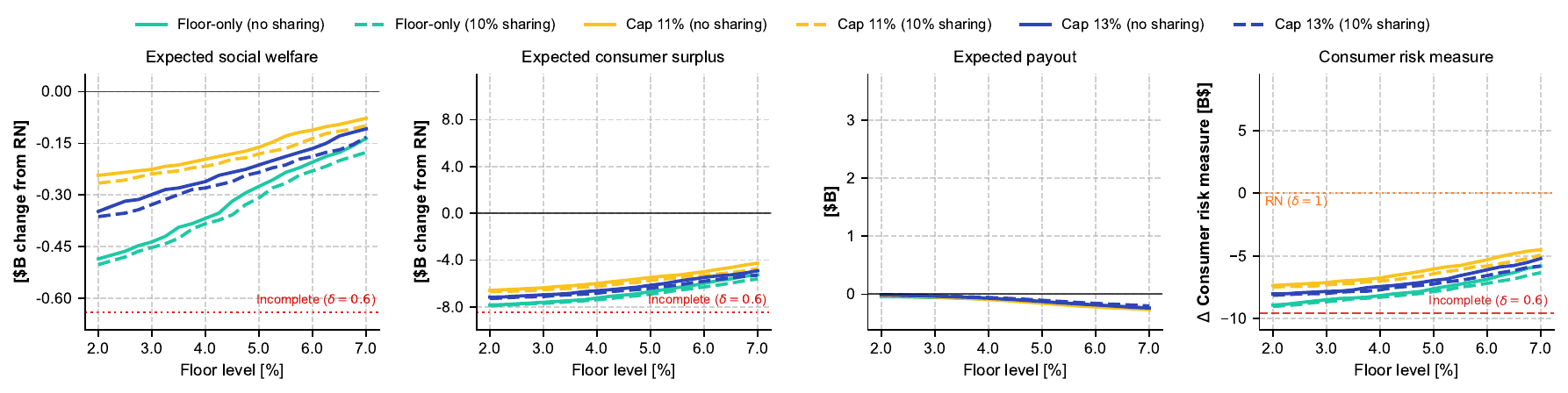}
\caption{{(left) change in expected social welfare from the risk-neutral level; (second from left) change in expected consumer welfare from the risk-neutral level; (second from right) expected total payouts. (right) risk-adjusted consumer welfare changes from the risk-neutral baseline. All figures are for selected LDES cap-and-floor designs at $\delta = 0.6$ and negotiated-premium C\&F.}}
\label{FIG:horizontal_negotiatedCF}
\end{figure*}

Fig.~\ref{fig:threshold_prices_vs_floor_normalized_delta_0} illustrates this coordination failure by comparing the range of premiums acceptable to consumers and investors for the investment levels obtained under zero-premium contracts. For floor-only contracts, consumers always require a positive premium because transfers occur only from consumers to investors. Introducing a cap, e.g., 11\% in this case study, creates bidirectional transfers, increasing the value of the contract to consumers and reducing the premium required for agreement. The point at which the negotiated premium crosses zero corresponds to the parameterization where negotiated and zero-premium contracts support the same LDES investment.

\begin{figure}[!htbp]
    \centering
    \includegraphics[width=0.95\linewidth]{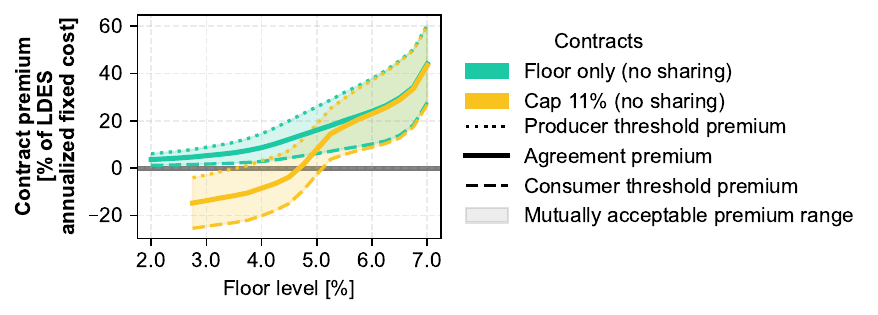}
    \caption{Mutually acceptable contract premium range for LDES capacity supported by zero-premium C\&F contracts at $\delta = 0.6$.}    \label{fig:threshold_prices_vs_floor_normalized_delta_0}
\end{figure}

Because negotiated premiums fail to internalize the investment externality, bargaining power becomes an important determinant of investment outcomes. Fig.~\ref{fig:capacity_and_cw_risk_diff_price_selection_envelopes} shows the equilibrium LDES capacity and risk-adjusted consumer welfare across the full range of mutually acceptable premiums. Premiums closer to the consumer participation threshold induce greater LDES investment and higher consumer welfare by providing stronger incentives for investment. Conversely, premiums close to the investor participation threshold recover the incomplete-market equilibrium, with negligible additional investment and limited welfare improvements. Thus, although all negotiated premiums satisfy individual rationality, they need not be socially efficient.

\begin{figure}[!htbp]
    \centering
    \includegraphics[width=\linewidth]{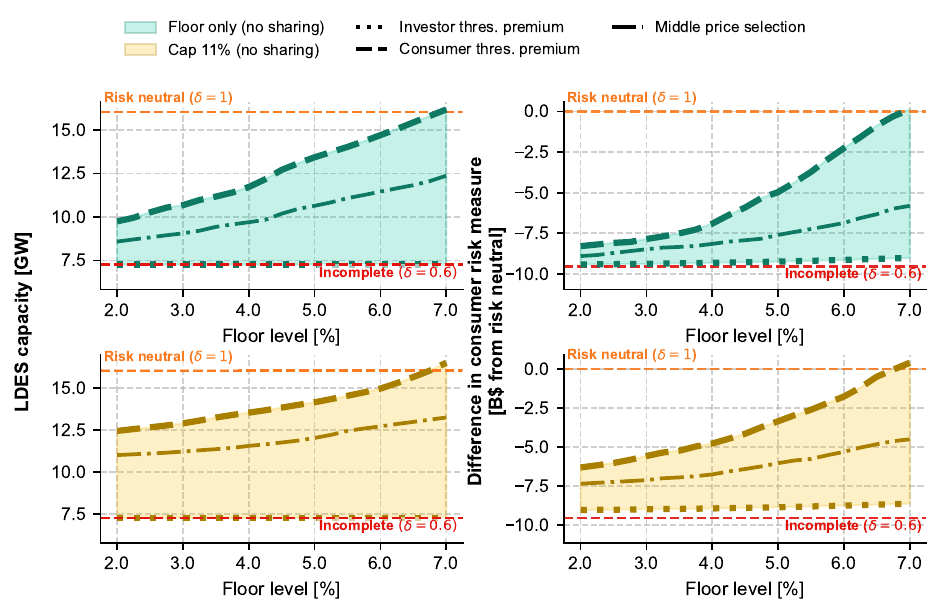}
    \caption{Equilibrium LDES capacity and consumer risk measure under negotiated contracts with contract premium lying inside the mutually acceptable range (between the consumer and producer threshold premiums).}
    
    \label{fig:capacity_and_cw_risk_diff_price_selection_envelopes}
\end{figure}

\section{Conclusions} \label{SEC:Conclusions}
This article investigates the effectiveness of Cap-and-Floor contracts for long-duration energy storage in electricity markets with incomplete financial markets. Using a stochastic equilibrium investment model with endogenous contract pricing, we evaluate how alternative cap-and-floor contract designs affect investment incentives, consumer welfare, and risk sharing between investors and consumers.

Our numerical results, for a stylized future Great British power system, show that market incompleteness substantially suppresses LDES investment by increasing investors' required risk premia.
We also show that the projected distribution of LDES revenues lies at the heart of Cap-and-Floor design, highlighting the importance of credible assumptions about future market design and operating rules when evaluating or implementing such schemes.

We find that centrally administered, zero-premium cap-and-floor contracts can effectively incentivize LDES investment by reducing downside revenue risk. However, the floor, cap, and revenue-sharing parameters interact strongly and therefore cannot be selected independently. Multiple combinations of these design parameters can deliver similar investment outcomes while implying different financing costs, consumer transfers, and operational incentives. Tight collars generally lead to lower financial risk premiums, but may weaken the asset's incentive to respond to wholesale market signals, requiring more sophisticated performance-incentive mechanisms to ensure efficient operation and avoid moral hazard. 

A second trade-off emerges when comparing zero-premium centrally administered contracts to contracts with negotiated premiums. Negotiated contracts achieve near-zero expected transfers between consumers and investors but generally fail to restore the risk-neutral investment level because decentralized bargaining does not account for the additional consumer benefits associated with higher LDES investment.

Overall, our results suggest that both contract design and institutional design play a central role in determining the effectiveness of cap-and-floor mechanisms for LDES. Policymakers should consider both to balance effects on investment incentives, consumer welfare, and the efficiency of risk sharing. Future work could extend this analysis by considering multiple LDES technologies, short-term uncertainty, and operating obligations for assets supported by cap-and-floor schemes.

\section*{Appendix 1: Nash Bargaining Solution to Contract Pricing}

Let $\rho$ be a coherent risk measure. Parties $A, B$ negotiate premium $P$ for a contract. Let $\rho_i^0, \rho_i^1$ denote baseline and post-contract risk measures. By translation invariance:
\begin{gather}
\rho_A(X_A^1 + P) = \rho_A^1 + P, \;\;
\rho_B(X_B^1 - P) = \rho_B^1 - P
\end{gather}

Define surpluses $g_A(P) = \rho_A^0 - \rho_A^1 - P =: \bar{P} - P$ and $g_B(P) = \rho_B^0 - \rho_B^1 + P =: P - \underline{P}$, where:
$\bar{P} := \rho_A^0 - \rho_A^1$, $\underline{P} := \rho_B^1 - \rho_B^0$

\begin{prop}
Assume $\underline{P} < \bar{P}$. The contract premium solution to Nash bargaining is $P^* = \frac{\bar{P} + \underline{P}}{2}$.
\end{prop}

\begin{proof}
Maximize $g_A(P) \cdot g_B(P) = (\bar{P} - P)(P - \underline{P})$. First-order condition:
\begin{gather}
\frac{d}{dP}[(\bar{P} - P)(P - \underline{P})] = 0
\end{gather}
Solving for P:     $P^* = \frac{\bar{P} + \underline{P}}{2}$
\end{proof}
\vspace{-0.65\baselineskip}
\bibliographystyle{IEEEtran}
\bibliography{references, mybib}
\vfill

\end{document}